\documentclass[a4paper,unpublished,onecolumn,11pt]{quantumarticle}
\pdfoutput=1

\usepackage{geometry,amssymb,latexsym,amsmath,graphicx}

\usepackage{stmaryrd}
\usepackage [autostyle, english = american]{csquotes}
\usepackage{hyperref}
\usepackage[capitalize]{cleveref}
\usepackage{amsthm}
\usepackage{tikz}
\MakeOuterQuote{"}
\usepackage{parskip}

\DeclareSymbolFont{matha}{OML}{txmi}{m}{it}
\DeclareMathSymbol{\varv}{\mathord}{matha}{118}

\definecolor{fillvertices}{RGB}{250,240,230}
\definecolor{cornflowerblue}{RGB}{100,149,237}

\usepackage{mathtools}

\newcommand{\commentaire}

\newcommand{\ket}[1]{  |{#1} \rangle} 
\newcommand{\bra}[1]{ \left \langle#1\right | }

\newcommand{\dloc}{\delta_\text{loc}}

\newcommand{\odd}[1]{\textup{Odd}({#1})}

\newcommand{\ls}{\leqslant}
\newcommand{\gs}{\geqslant}
\newcommand{\sm}{\setminus}

\newtheorem*{theorem*}{Theorem}
\newtheorem{definition}{Definition}
\newtheorem{lemma}{Lemma}
	
\newtheorem{proposition}{Proposition}
\newtheorem{corollary}{Corollary}
\newtheorem{remark}{Remark}

\title{Small $k$-pairable states}
\author{Nathan Claudet}
\affiliation{Inria Mocqua, LORIA, CNRS, Universit\'e de Lorraine,  F-54000 Nancy, France}
\author{ Mehdi Mhalla}
\affiliation{Universit\'e Grenoble Alpes, CNRS, Grenoble INP, LIG, F-38000 Grenoble, France}
\author{ Simon Perdrix}
\affiliation{Inria Mocqua, LORIA, CNRS, Universit\'e de Lorraine,  F-54000 Nancy, France}

\date{}
\begin{document}

\maketitle

\begin{abstract}
    A $k$-pairable $n$-qubit state is a resource state that allows Local Operations and Classical Communication (LOCC) protocols to generate EPR-pairs among any $k$-disjoint pairs of the $n$ qubits. Bravyi et al.~introduced a family of $k$-pairable $n$-qubit states, where $n$ grows exponentially with $k$. Our primary contribution is to establish the existence of `small' pairable quantum states. Specifically, we present a family of $k$-pairable $n$-qubit graph states, where $n$ is polynomial in $k$, namely $n=O(k^3\ln^3k)$. Our construction relies on probabilistic methods.

    Furthermore, we provide an upper bound  on the pairability of any arbitrary quantum state based on the support of any local unitary transformation that has the shared state as a fixed point. This lower bound implies that the pairability of a graph state is at most half of the minimum degree up to local complementation of the underlying graph, i.e., $k(\ket G)\ls \lceil \delta_{loc}(G)/2\rceil$.

    We also investigate the related combinatorial problem of $k$-vertex-minor universality: a graph $G$ is $k$-vertex-minor universal if any graph on any  $k$ of its vertices is a vertex-minor of $G$. When a graph is $2k$-vertex-minor universal, the corresponding graph state is $k$-pairable. More precisely, one can create not only EPR-pairs but also any stabilizer state on any $2k$ qubits through local operations and classical communication. We establish the existence of $k$-vertex-minor universal graphs of order $O(k^4 \ln k)$.

    Finally, we explore a natural extension of pairability in the presence of errors or malicious parties and show that vertex-minor universality ensures a robust form of pairability.
\end{abstract}

\section{Introduction}

In the realm of quantum communication networks, we often rely on classical communication along with shared entanglement. Since classical communication cannot create entanglement, we must rely on pre-existing entangled states to perform non-trivial operations. For example, an EPR-pair $\frac{1}{\sqrt{2}}(\ket{00} + \ket{11})$ shared between two parties allows the quantum teleportation of a qubit \cite{PhysRevLett.70.1895}. In this context, a highly pertinent problem is to explore which resource states enable a group of $n$ parties, equipped with the capability of employing Local Operations and Classical Communication (LOCC), to create entangled EPR pairs among any $k$ pairs of qubits. It is only recently that Bravyi et al.~addressed this fundamental inquiry and provided both upper and lower bounds for what they call the $k$-pairability of quantum states, in terms of the number of parties and the number of qubits per party needed for a quantum state to be $k$-pairable \cite{DeWolf22}. However, before their work, numerous variations of this problem had surfaced in the literature, some in the context of entanglement routing \cite{schoute2016shortcuts,hahn2019quantum,pant2019routing}, and some about problems that can be described as variants of $k$-pairability, for example to prepare resource states by clustering and merging  \cite{Miguel_Ramiro_2023}; starting from a particular network \cite{DSL:multipoint,Contreras_Tejada_2022}; creating one EPR pair hidden from other parties \cite{illianoetal}; studying the complexity and providing algorithms for generating EPR pairs within a predetermined set \cite{dahlberg2020transforming,DHW:howtotransform}; or taking into account the cost of distributing a graph state in terms of EPR pairs \cite{meignant2019distributing,fischer2021distributing} (see \cite{DeWolf22} for a more detailed review).

Formally, an $n$-party state $\ket \psi$ is said to be \textit{$k$-pairable} if, for every $k$ disjoint pairs of parties  $\{a_1, b_1\},\ldots,\{a_k, b_k\}$, there exists a LOCC protocol that starts with $\ket \psi$ and ends up with a state where each of those $k$ pairs of parties shares an EPR-pair. Bravyi et al.~studied $n$-party states in the case where each party holds $m$ qubits, with $m$ ranging from 1 to $\ln(n)$. In the case where each party holds at least $10$ qubits, they showed the existence of $k$-pairable states where $k$ is of the order of $n/ \text{polylog}(n)$, which is nearly optimal when $m$ is constant.
They also showed that if one allows a logarithmic number of qubits per party, then there exist $k$-pairable states with $k = n/ 2$.

In the present paper, we focus on the scenario that is both the most natural and challenging: when each party possesses precisely one qubit, i.e.~$m=1$. Bravyi et al.~proved some results in this particular case. First, using Reed-Muller codes,  they were able to construct, for any $k$, a $k$-pairable state of size  exponential in $k$, namely $n = 2^{3k}$,  leaving  the existence of a $k$-pairable state of size $n=poly(k)$ as an open problem. They found a 2-pairable graph state of size $10$ and proved that there exists no stabilizer state on less than 10 qubits that is 2-pairable using LOCC protocol based on Pauli measurements. 

Our contributions rely on the graph state formalism and the ability to characterize properties of quantum states using tools from graph theory. In particular, the pairability of a graph state is related to the standard notion of vertex-minors. A graph $H$ is a vertex-minor of $G$ if one can transform $G$ into $H$ by means of local complementations\footnote{Local complementation according to a vertex $u$ consists in toggling the edges in the neighbourhood of $u$.} and vertex deletions. If $H$ is a vertex-minor of a $G$ then the graph state $\ket{H}$ can be obtained from $\ket G$ using only single-qubit Clifford operations, single-qubit Pauli measurements and classical communications (we call these protocols CLOCC\footnote{In \cite{dahlberg2020transforming} this fragment of operations is called LC + LPM + CC, and in \cite{DeWolf22} this corresponds to "LOCC protocols based on Pauli measurements".} for Clifford LOCC). Dahlberg, Helsen, and Wehner proved that the converse is also true when $H$ has no isolated vertices \cite{DWH:transfo}. In \cite{dahlberg2020transforming}, they proved that it is NP-complete to decide whether a graph state can be transformed into a set of EPR-pairs on specific qubits using CLOCC protocols. In \cite{DHW:howtotransform} they showed that it is also NP-complete to decide whether a graph state can be transformed into another one using CLOCC protocols. 

We prove here the existence of an infinite family of $k$-pairable $n$-qubit graph states, where the number of qubits $n$ is polynomial in $k$ (specifically $n=O(k^3\ln^3k)$), while the construction from Bravyi et al.~results in  $k$-pairable states with an exponential number of qubits. For this purpose, we first point out that a graph state $\ket G$ is $k$-pairable if any matching of size $k$ (graph consisting of a set of $k$ disjoint edges) is a vertex-minor of $G$. We then use probabilistic methods to prove the existence of such $k$-pairable graph states with a number of qubits polynomial in $k$.

We also provide an upper bound on the $k$-pairability of a graph state $\ket G$, namely $k$ is at most half of the local minimum degree $\dloc(G)$. The local minimum degree \cite{CattaneoP15} refers to the minimum degree of a vertex of $G$ subject to any sequence of local complementation. The local minimum degree is related to the size of the smallest local set in a graph \cite{Perdrix06}, which we put to use here. Note, however, that the decision problem associated with the  computation of the local minimum degree of a graph has been proven to be NP-complete and hard to approximate \cite{Javelle12}. 

This bound is not directly comparable to the bound proposed by Bravyi et al.~\cite{DeWolf22}, which, roughly speaking, implies that $k=O(n\frac{\ln\ln n}{\ln n})$. The new bound is significantly better in certain cases; for instance, it directly implies that graph states whose underlying graph has a vertex of constant degree have constant pairability (as opposed to almost linear). However, it is worth noting that there are graphs with a local minimum degree linear in their number of vertices \cite{Javelle12}, although no explicit construction for such graphs is known to our knowledge. In such cases, the bound provided by \cite{DeWolf22} can be better than the one based on the local minimum degree by a logarithmic factor. In the process of proving this bound on the pairability of graph state, we prove a bound on the pairability of arbitrary quantum states, which depends on the support of local unitaries that the state is a fixed point of.

From a combinatorial perspective, it is natural to consider graphs that contain any graph of a given order as a vertex-minor, rather than solely focusing on matchings. This leads us to introduce the notion of vertex-minor universality. We say that a graph $G$ is \textit{$k$-vertex-minor universal} if any graph on any $k$ of its vertices is a vertex-minor of $G$. If a graph $G$ is $k$-vertex-minor universal then one can create any stabilizer state on any $k$ qubits of the graph state $\ket G$ by CLOCC protocols. As a consequence, if $G$ is $2k$-vertex-minor universal then $\ket G$ is $k$-pairable. We prove the existence of an infinite family of $k$-vertex-minor universal graphs, where the number of vertices $n$ is polynomial in $k$, namely $n=O(k^4 \ln k)$ using probabilistic methods. Moreover,  a counting argument implies that $n$ is at least quadratic in $k$, and we show that $k$ is upper-bounded by the local minimum degree. Furthermore, we present minimal examples of graphs that are $2$-, $3$-, and $4$-vertex-minor universal.

In the context of quantum networks, it is important to study robustness to errors or malicious parties. For this purpose, we introduce a robust version of  pairability. We say that  a $k$-pairable state $\ket \psi$ on a set  of qubits is $m$-robust  if for any set of size at most $m$ of malicious partners, the  trusted partners can create $k$ EPR pairs among any $2k$ of them with an LOCC protocol. We prove that vertex-minor universality ensures robust pairability. 

In \cref{sec:def}, we define the pairability of a quantum state through the graph state formalism. In \cref{sec:dloc}, we prove upper bounds on the pairability of a graph state, using in particular the local minimum degree. In \cref{sec:ex_pairable}, we prove the existence of polynomial-size $k$-pairable graph states. In \cref{sec:ex_universal}, we introduce the notion of $k$-vertex-minor universality and prove the existence of polynomial-size $k$-vertex-minor universal graphs. We also provide a table with a few examples of $k$-pairable states and $k$-vertex-minor universal graphs for small values of $k$.  Finally, in \cref{sec:robust}, we define robust pairability and show how vertex-minor universality implies robust pairability.

\section{Quantum state pairability}
\label{sec:def}

We review in this section basic definitions of $k$-pairability and exemplify this concept with existing and new illustrating instances. 
We extensively use the graph state formalism \cite{Hein06}, which is a standard family of quantum states that can be represented using simple undirected graphs. Given a graph $G=(V,E)$, the corresponding graph state $\ket G$ is the $|V|$-qubit state: $$\ket G=\frac{1}{\sqrt 2^{|V|}}\sum_{x\in 2^V} (-1)^{|G[x]|}\ket x$$ where $|G[x]|$ is the size (number of edges) of the subgraph induced by $x$ $^($\footnote{With a slight abuse of notation we identify a subset (say $x=\{u_2,u_4\}$) of the set of qubits $V=\{u_1,\ldots u_5\}$ with its characteristic binary word  ($x=01010$).}$^)$. 

A graph state $\ket G$ can be prepared as follows: initialize every qubit in  $\ket + = \frac{\ket 0+\ket 1}{\sqrt 2}$ then apply for each edge of the graph a $CZ:\ket{ab}\mapsto (-1)^{ab}\ket{ab}$ gate on the corresponding pair of qubits. Notice that the graph state $\ket G$ is the unique quantum state (up to a global phase) which is, for every vertex $u\in V$,  a fixed point of the Pauli operator $X_uZ_{N(u)}$ $^($\footnote{It consists in applying $X:\ket a\mapsto \ket {1{-}a}$ on $u$ and $Z:\ket a \mapsto (-1)^a\ket a$ on each of its neighbours in $G$.}$^)$.  

A matching $\pi$ of size $k$  and order $n$ is a graph on $n$ vertices made of $k\ls \frac n 2$ pairwise non-adjacent edges.  The associated  graph state $\ket \pi$ is nothing but $k$ pairs of maximally entangled qubits together with $n-2k$ fully separable qubits. Notice that there are $\frac{n! 2^{-k}}{k!(n-2k)!}$ matchings of size $k$ and order $n$.

A protocol that transforms an $n$-party state $\ket \psi$ into a state $\ket {\varphi}$ is LOCC (resp. CLOCC) if it uses only local operations (resp. local Clifford unitaries and Pauli measurements)  and classical communications between the parties. In this paper we consider only protocols where each party is made of a single qubit.

\begin{definition}
\label{def:pairability}
An $n$-qubit quantum state $\ket \psi$ is \emph{$k$-pairable} if for every matching $\pi$ of size $k$  and order $n$ there is an LOCC protocol transforming $\ket \psi$ into $\ket \pi$.
\end{definition}

\begin{remark}
Here we only consider the case where each party is made of a single qubit, please refer to \cite{DeWolf22} for a more general definition. Notice that in \cite{DeWolf22}, $k$-pairability  is defined as the ability to produce any $k$ EPR pairs, which is equivalent to \cref{def:pairability} since given a matching $\pi$ of size $k$, $\ket \pi$ is, up to local Clifford unitaries, nothing but $k$ EPR pairs together with $n-2k$ qubits in the $\ket 0$-state. 
\end{remark}

For instance the GHZ state $\frac {\ket {0^n}+\ket {1^n}}{\sqrt 2}$ is $1$-pairable. More generally, a graph state $\ket G$ is $1$-pairable if and only if $G$ is connected. Bravyi et al. showed that the graph state -- they call the ``wheel'' -- on $10$ qubits is $2$-pairable (\cref{fig:wheel}.Left).

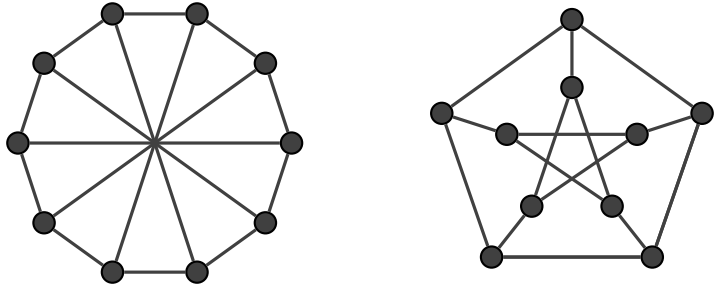
\begin{figure}
\centering
$
\begin{tikzpicture}[scale = 0.3]

\begin{scope}[shift={(26,0)},every node/.style={circle,minimum size=8pt,thick,draw,fill=darkgray, inner sep = 0pt}, scale = 2]
    \node (U0) at (-0.927,2.853) {};
    \node (U1) at (0.927,2.853) {};
    \node (U2) at (2.427,1.763) {};
    \node (U3) at (3,0) {};
    \node (U4) at (2.427,-1.763) {};
    \node (U5) at (0.927,-2.853) {}; 
    \node (U6) at (-0.927,-2.853) {};  
    \node (U7) at (-2.427,-1.763) {};
    \node (U8) at (-3,0) {};
    \node (U9) at (-2.427,1.763) {};  
\end{scope}
\begin{scope}[every node/.style={},
                every edge/.style={draw=darkgray,very thick}]              
    \path [-] (U0) edge node {} (U1);
    \path [-] (U1) edge node {} (U2);
    \path [-] (U2) edge node {} (U3);
    \path [-] (U3) edge node {} (U4);
    \path [-] (U4) edge node {} (U5);
    \path [-] (U5) edge node {} (U6);
    \path [-] (U6) edge node {} (U7);
    \path [-] (U7) edge node {} (U8);
    \path [-] (U8) edge node {} (U9);
    \path [-] (U9) edge node {} (U0);
    \path [-] (U0) edge   node {} (U5);
    \path [-] (U1) edge   node {} (U6);
    \path [-] (U2) edge   node {} (U7);
    \path [-] (U3) edge  node {} (U8);
    \path [-] (U4) edge   node {} (U9);
\end{scope}
\end{tikzpicture}\qquad\qquad\raisebox{0.2cm}{
    \begin{tikzpicture}[scale = 0.3]
    
    \begin{scope}[shift={(26,0)},every node/.style={circle,minimum size=8pt,thick,draw,fill=darkgray}, scale = 2, inner sep = 0pt]
        \node (U0) at (-1.427,0.464) {};
        \node (U1) at (0,1.5) {};
        \node (U2) at (1.427,0.464) {};
        \node (U3) at (0.881,-1.124) {};
        \node (U4) at (-0.881,-1.124) {};
        \node (U5) at (-1.427*2,0.464*2) {}; 
        \node (U6) at (0,1.5*2) {};  
        \node (U7) at (1.427*2,0.464*2) {};
        \node (U8) at (0.881*2,-1.124*2) {};
        \node (U9) at (-0.881*2,-1.124*2) {};  
    \end{scope}
    \begin{scope}[every node/.style={},
                    every edge/.style={draw=darkgray,very thick}]              
        \path [-] (U0) edge node {} (U3);
        \path [-] (U3) edge node {} (U1);
        \path [-] (U1) edge node {} (U4);
        \path [-] (U4) edge node {} (U2);
        \path [-] (U2) edge node {} (U0);

        \path [-] (U5) edge node {} (U6);
        \path [-] (U6) edge node {} (U7);
        \path [-] (U7) edge node {} (U8);
        \path [-] (U8) edge node {} (U9);
        \path [-] (U9) edge node {} (U5);

        \path [-] (U7) edge node {} (U8);
        \path [-] (U8) edge node {} (U9);
        \path [-] (U0) edge node {} (U5);
        \path [-] (U1) edge node {} (U6);
        \path [-] (U2) edge node {} (U7);
        \path [-] (U3) edge node {} (U8);
        \path [-] (U4) edge node {} (U9);
    
    \end{scope}
    \end{tikzpicture}}$

\caption{(Left) ``Wheel'' graph. (Right) Petersen graph}
\label{fig:wheel}
\end{figure}

This example is somehow \emph{minimum} in the sense that there is no graph state on at most $9$ qubits that are $2$-pairable using a CLOCC protocol. We introduce a new example of $2$-pairable  state on 10 qubits, namely the graph state associated with the Petersen graph  (\cref{fig:wheel}.Right). We provide also an example of $3$-pairable graph state associated with the 29-Paley graph\footnote{Given a prime $q = 1 \bmod 4$, the $q$-Paley graph is a graph which vertices are elements of the finite field $GF(q)$, such that two vertices share an edge if their difference is a square in $GF(q)$.} (on 29 vertices) which improves on the $3$-pairable state on 32 qubits introduced in \cite{DeWolf22}  (see \cref{fig:table}).

Since pairability is defined by means of LOCC protocols, two quantum states that are equal up to local unitaries (LU-equivalent for short) have the same pairability.  

A useful graph transformation to describe equivalent graph states by local Clifford unitaries is local complementation. Given a graph  $G$, a local complementation according to a given vertex $u$ consists in complementing the subgraph induced by the neighbourhood of $u$, leading to the graph   $G\star u= G\Delta K_{N(u)}$ where $\Delta$ is the symmetric difference and $K_A$ is the complete graph on the vertices of $A$. It has been shown in  \cite{VandenNest04} that if two graphs are equivalent up to local complementation then the corresponding graph states are LU-equivalent and hence have the same pairability.

Notice that the graph states associated with the ``wheel'' and the Petersen graph are not LU-equivalent. Notice also that the Petersen graph has already been pointed out as the smallest examples in various properties related to local complementations \cite{Bouchet1993,Hein06}.

\section{Upper bounds}
\label{sec:dloc}

It is known that the pairability of quantum states is sublinear in their number of qubits, roughly speaking $k=O(n\frac{\ln \ln n}{\ln n})$ \cite{DeWolf22}. In this section, we provide an alternative upper bound on the pairability of arbitrary quantum states, depending on the support of local unitaries that the state is a fixed point of. When applied to graph states, this upper bound can be expressed in terms of graph  parameters: local minimum degree (\cref{cor:deltaloc}) and vertex cover number (\cref{cor:vertexcover}). Recall that the support of a local unitary $U=\bigotimes_{v\in V} U_v$ is the set of qubits on which $U$ acts not trivially: $supp(U) = \{v \in V ~|~ U_v \not\propto I\}$ $^($\footnote{We note $U \propto V$ when $U = e^{i\phi} V$ for some angle $\phi$, and $U \not\propto V$ otherwise.}$^)$. 

\begin{proposition}
    \label{prop:upperbound}
   If $\ket{\psi}$ is $k$-pairable and $U$ is a local unitary such that  $U\neq I$ and $U\ket{\psi}=\ket \psi$, then $2k\ls |supp(U)|$.
\end{proposition}

\begin{proof}
By contradiction, assume there exist a $k$-pairable state $\ket\psi$ of a register $V$ of qubits and a local unitary $U$~s.t.~$U\ket{\psi}=\ket \psi$ and $2k>|supp(U)|>0$. 
We consider a matching $\pi$ such that every vertex in the support of $U$ is covered by an edge and at least one edge has one vertex inside and one vertex outside $supp(U)$, i.e.  $\pi=(V,E)$ with $E= \{(u_i,v_i)\}_{i=1\ldots k}$, 
$u_1\in supp(U)$ and $v_1\notin supp(U)$, and $supp(U)\subseteq C$ where $C:=\{u_i\}_{i=1\ldots k}\cup \{v_i\}_{i=1\ldots k}$ is the set of vertices covered by the matching. By hypothesis, the state $\ket \pi$ can be obtained from the state $\ket \psi$ by an LOCC protocol (with non-zero probability). Any LOCC protocol can be described by a completely trace-preserving map, with separable Kraus operators. So there exists a product Kraus operator $M = M_1 \otimes M_2 \otimes\dots \otimes M_n$ such that $M \ket \psi = c_{\pi} \ket \pi$ for some $c_{\pi} \neq 0$. Notice that $M_u$ is invertible when $u\in C$.  Indeed, let $\ket{\varphi_0}$ and $\ket {\varphi_1}$ be two independent eigenvectors of $M_u$. We have $\ket \psi = \ket{\varphi_0}_u\otimes\ket{\psi_0}_{V \sm u} + \ket{{\varphi_1}}_u\otimes\ket{{\psi_1}}_{V \sm u}$. 

So $M\ket \psi  = M_u\ket{\varphi_0}_u\otimes M_{V\sm u}\ket{\psi_0}_{V \sm u} + M_u\ket{{\varphi_1}}_u\otimes M_{V\sm u}\ket{{\psi_1}}_{V \sm u} = c_{\pi} \ket \pi$. Since $u$ is entangled in $\ket \pi$ and $c_\pi\neq 0$, we have $M_u\ket{\varphi_0}_u\neq 0$ and $M_u\ket{\varphi_1}_u\neq 0$, so $M_u$ is invertible. As a consequence, $M_{C}:= (\bigotimes_{u\in  C}M_u) \otimes I_{V\setminus  C}$ is invertible. We consider also $M_{\overline C}:= M_{C}^{-1}M$ which commutes with $U$ by construction. We show that $\ket \pi$ is a fixed point of $M_{C}UM_{\overline  C}^{-1}$. Indeed, $c_\pi \ket \pi= MU\ket \psi = 
M_{C}M_{\overline C}U\ket \psi =M_{ C}UM_{\overline C} \ket \psi= M_{C}UM_{C}^{-1}M_{C}  M_{\overline C} \ket \psi = M_{C}UM_{C}^{-1}M \ket \psi\\=  c_\pi M_{C}UM_{C}^{-1} \ket \pi$. 

It implies on the pair of qubits $u_1(\in supp(U))$, $v_1(\notin supp(U))$, that $\ket{\Phi}_{u_1,v_1}$ is an eigenstate of $M_{u_1}U_{u_1}M_{u_1}^{-1}\otimes I_{v_1}$, so there exists $\lambda\in \mathbb C$ s.t.~$M_{u_1}U_{u_1}M_{u_1}^{-1}\ket 0 =\lambda \ket 0$ and  $M_{u_1}U_{u_1}M_{u_1}^{-1}\ket 1 =\lambda \ket 1$, as a consequence $M_{u_1}U_{u_1}M_{u_1}^{-1} = \lambda I$, so $U_{u_1}=\lambda I$. Since $U_{u_1}$ is unitary, we have $U_{u_1}\propto I$ which contradicts the hypothesis $u_1\in supp(U)$.  
\end{proof}

In the particular case of graph states, we can express this bound in terms of an already-studied graph parameter on the corresponding graph, namely the local minimum degree, which is the minimum degree up to local complementation: $\delta_{loc}(G):= \min_{G\equiv_{LC} H} \delta(H)$ where $G \equiv_{LC} H$ means that $G$ can be transformed into $H$ using a series of local complementations. 
Indeed, the minimum support of a Pauli operator  stabilizing a graph state $\ket G$ is equal to $\delta_{loc}(G)+1$  \cite{Perdrix06}. 
It leads to the following upper bound:

\begin{corollary}
    \label{cor:deltaloc}
    A graph state $\ket G$ is not $\left( \left\lceil \dloc(G)/ 2\right\rceil + 1\right)$-pairable.
\end{corollary}

\begin{proof}
Since for any $u\in V$, $\ket G$ is a fixed point of $X_uZ_{N(u)}$, $\ket G$ is also a fixed point of $\prod_{u\in D}X_uZ_{N(u)} = \pm X_DZ_{\odd D}$ for any $D\subseteq V$, where $\odd D:= \{v\in V~|~|N(v)\cap D|=1\bmod 2\}$. 
The support $D\cup \odd D$ of $X_DZ_{\odd D}$ is called a local set. The minimum size of a non-empty local set is known to be equal to $\delta_{loc}(G)+1$ \cite{Perdrix06}. Thus, according to \cref{prop:upperbound}, $\ket G$ is not  $\left( \left\lceil \dloc(G)/ 2\right\rceil + 1\right)$-pairable.
\end{proof}

This bound is tight, for instance the two graphs of \cref{fig:wheel} have local minimum degree $3$ and the corresponding graph states are both $2$-pairable. We believe that \cref{cor:deltaloc} generally provides  a useful upper bound on the pairability of a  graph state as it is challenging to find a constructive family of graphs with large local minimum degree. The family of hypercubes, for example, has a logarithmic local minimum degree \cite{Perdrix06}. Paley graphs provide, up to our knowledge, the best\footnote{`Best' here means with the largest ratio local minimum degree divided by order of the graph.} known constructive family with graphs of order quadratic in their local minimum degree (see \cref{fig:table} for small Paley graphs).

However, it has been proved, using non-constructive probabilistic methods, that  for any (large enough) $n$  there is exists a graph of order $n$ and local minimum degree greater than $0.189n$ \cite{Javelle12}, hence the upper bound provided by \cref{cor:deltaloc} is not tight for such graphs as it is known that the pairability of a quantum state is sublinear in its number of qubits ($k=O(n\frac{\ln \ln n}{\ln n})$) \cite{DeWolf22}.

Finally, the local minimum degree is related to the vertex cover number\footnote{I.e. the size of the smallest set $S$ such that if $u$ and $v$ share an edge, then $u \in S$ or $v \in S$.} $\tau(G)$. Namely, the pairability of a graph state is at most a quarter of its vertex cover number up to logarithmic terms: 

\begin{corollary}
    \label{cor:vertexcover}
    A graph state $\ket G$ is not $\left( \left\lceil \frac{\tau(G) + \log_2(\tau(G))}{4}\right\rceil + 1\right)$-pairable.
\end{corollary}

\begin{proof}
It is known that $2 \dloc(G) \ls \tau(G) + \log_2(\tau(G)) + 1$ \cite{CattaneoP15}, and that we can remove the constant term when $\tau(G)\neq 1$. When $\tau(G)=1$, we have  $\left\lceil \frac{\tau(G) + \log_2(\tau(G))}{4}\right\rceil = \left\lceil \frac{1 + 0}{4}\right\rceil = 1 = \left\lceil \frac{1 + 0 + 1}{4}\right\rceil = \left\lceil \frac{\tau(G) + \log_2(\tau(G))+1}{4}\right\rceil$. 
\end{proof}

\section{Existence of a polynomial-size $k$-pairable graph state}
\label{sec:ex_pairable}

In this section, we prove the existence of an infinite family of $k$-pairable graph states with a number of qubits that is polynomial in $k$. For this purpose we  describe a sufficient vertex-minor-based condition on a graph $G$ for the corresponding graph state $\ket G$ to be $k$-pairable.  

\emph{Vertex-minor} is a standard notion in graph theory \cite{OUM200579, COURCELLE200791}: Given two graphs $H=(V_H,E_H)$ and $G=(V_G,E_G)$  such that $V_H\subseteq V_G$, $H$ is a vertex-minor of $G$ when $H$ can be obtained from $G$ by means of local complementations and vertex deletions.  It is well known that one can actually transform $H$ into $G$ by applying first the local complementations and then the vertex deletions.  In other words, $H$ is the subgraph induced by $V_H$ in the graph $G\star u_1\star u_2\ldots \star u_m$ for some sequence of vertices $u_1, \ldots , u_m$. 

If a graph $H$ is a vertex-minor of a graph $G$ then the graph state $\ket{H}$ can be obtained from $\ket G$ by CLOCC protocols, the converse is also true when $H$ has no isolated vertices \cite{DWH:transfo}. Notice that a perfect matching\footnote{A perfect matching is a graph where each vertex appears in exactly one edge.} has by definition no isolated vertices. Then,

\begin{proposition}
    \label{prop:cond_k_pairaibility}
    A graph state $\ket G$ is $k$-pairable by CLOCC protocols if and only if  any perfect matching on any $2k$ vertices is a vertex-minor of $G$.
\end{proposition}

To prove the existence of `small' graphs which admit any perfect matching on any $2k$ vertices as vertex-minors we use probabilistic methods. 

We consider a  random Erd\"os Rényi graph $G(n,p)$ of order $n$ and such that each edge is included in the graph with probability $p$, independently of every other edge. The objective is to choose  appropriate values of $p$ and $n$ to guarantee that the random graph admits, with a non-zero probability, every perfect matching of size $k$  as a vertex-minor.


To do so, we are going to use the following lemma, which derives directly from the union bound:


\begin{lemma}
    \label{lemma:unionbound}
    Let $\mathcal{A} = \{A_1, \dots , A_d\}$ be a set of bad events in an arbitrary probability space. If for all $A_i$, $Pr(A_i) \le  p$ and $dp<1$, then with a non-zero probability, none of the bad events occur: $Pr(\overline{A_1}, \dots , \overline{A_d}) > 0$.
\end{lemma}

\begin{proof}
    $Pr(\overline{A_1}, \dots , \overline{A_d}) = 1 - Pr(A_1 \cup \dots \cup A_d) \ge 1 - \sum_{i \in \llbracket 1, d \rrbracket} Pr(A_i) \ge 1 - dp>0$.    
\end{proof}

In our context, a bad event is when the strategy described below fails to induce a given perfect matching on a given set of vertices.  Intuitively, in order to induce a given perfect matching on a fixed set $K$ of vertices, one has to toggle some edges, say $r$ edges. Such an edge $(a,b)$ can be toggled by means of a local complementation on a vertex $u_{a,b}\notin K$  if $u_{a,b}$ is connected to both $a$ and $b$ but none of the other vertices of $K$. To guarantee that each of the $r$ edges can be toggled independently, it is also desirable that the corresponding $r$ vertices $u_{a,b}$ form an independent set (see \cref{fig:badevent}). 

We first prove a technical lemma to upper bound the probability of the bad event that such a configuration does not exist in a random graph:  
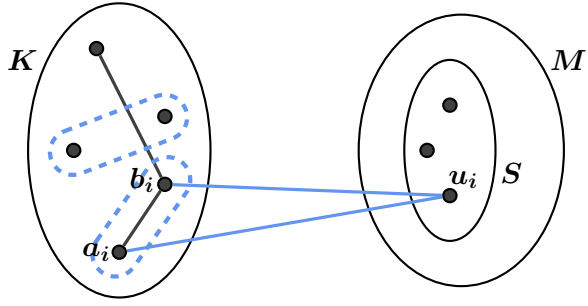
\begin{figure}
\centering
\begin{tikzpicture}[scale = 0.3]

\draw[thick] (0,0) ellipse (4 and 6.5);

\draw[thick] (15,0) ellipse (4.5 and 6);

\draw[thick] (14.5,0) ellipse (2 and 4);

\draw (-4.25,4) node(){$\boldsymbol{K}$};
\draw (19.7,4) node(){$\boldsymbol{M}$};
\draw (17.2,-1) node(){$\boldsymbol{S}$};

\begin{scope}[shift={(0,0)},every node/.style={circle,minimum size=5pt,thick,draw,fill=darkgray, inner sep = 0pt}]    
    \node (U1) at (-1,4.5) {};
    \node (U2) at (2,1.5) {};
    \node (U3) at (-2,0) {};
    \node (U4) at (2,-1.5) {};
    \node (U5) at (0,-4.5) {};
\end{scope}
\begin{scope}[shift={(13.5,0)},every node/.style={circle,minimum size=5pt,thick,draw,fill=darkgray, inner sep = 0pt}]    
    \node (U6) at (0,0) {};
    \node (U7) at (1,2) {};
    \node (U8) at (1,-2) {};

\end{scope}
\begin{scope}[every node/.style={},
                every edge/.style={draw=darkgray,very thick}]              
    \path [-] (U1) edge node {} (U4);
    \path [-] (U4) edge node {} (U5);
\end{scope}
\begin{scope}[every node/.style={},
    every edge/.style={draw=cornflowerblue,very thick}]              
    \path [-] (U8) edge node {} (U4);
    \path [-] (U8) edge node {} (U5);
\end{scope}

\begin{scope}[rotate=0,shift={(-1,-4)},rotate=-35]
    \draw[cornflowerblue, ultra thick, dashed] (0,0) -- (0,4) arc(180:0:1) -- (2,0) arc(0:-180:1);
\end{scope}

\begin{scope}[rotate=0,shift={(-2.4,0.85)},rotate=-69]
    \draw[cornflowerblue, ultra thick, dashed] (0,0) -- (0,4.3) arc(180:0:1) -- (2,0) arc(0:-180:1);
\end{scope}

\draw (15.1,-1.3) node(){$\boldsymbol{u_i}$};
\draw (-1,-4.5) node(){$\boldsymbol{a_i}$};
\draw (1,-1.3) node(){$\boldsymbol{b_i}$};

\end{tikzpicture}

\caption{Illustration of the strategy of \cref{lemma:probaS}: we want to construct an independent set $S$ of vertices in $M$ such that for any pair $\{a_i,b_i\} \in R$, there is a $u_i \in S$ s.t. $N_{K\cup S}(u_i)= \{a_i,b_i\}$.}
\label{fig:badevent}
\end{figure}

\begin{lemma}
    Given a random graph $G(n,p)$, let $K$ (of size $k$) and $M$ (of size $m$) be two non-intersecting subsets of vertices, and $R\subseteq {K\choose 2}$  be a set of $r$ pairs of vertices in $K$. The probability that for any independent set  $S\subseteq M$, $\exists \{a,b\}\in R$, s.t.~$\forall u\in S$, $N(u)\cap K\neq \{a,b\}$, is upper bounded by $r e^{-(m-r)p^2(1-p)^{k+r-2}}$.
    \label{lemma:probaS}    
\end{lemma}

\begin{proof}
Let $R = \{\{a_1,b_1\}, \ldots \{a_r,b_r\}\}$. We consider the complementary event $\overline B$: $\exists S\subseteq M \text{~s.t.~} \forall i \in \llbracket 1, r \rrbracket, \exists u_i\in S~s.t.~N_{K\cup S}(u_i)= \{a_i,b_i\}$, where  $N_A(u):=N(u)\cap A$. We describe in the following a greedy  algorithm to find such a set  $S$, and then upper bound $Pr(B)$ by the probability that the algorithm fails. First initialize $S$ as the empty set, and consider an arbitrary ordering on the vertices of $M$. We will consider vertices $u$ of $M$ one after the other. At each step, if $\exists i \in \llbracket 1, r \rrbracket$, $N_{K\cup S}(u)= \{a_i,b_i\}$, then we add $u$ to $S$, and we remove $\{a_i,b_i\}$ from $R$. When  $R$ is empty, we are done.  We note $p(m,r,s)$ the probability that the algorithm fails if we start with a set $S$ of size $s$ (then the probability that the algorithm fails in general is $p(m,r,0)$). We have $p(m,r,s)=1$ when $r>m$ and $p(m,0,s) =0$. More generally, we show that:  
$$p(m,r,s) = rp^2(1-p)^{k+s-2}p(m-1,r-1,s+1) + \left(1 - rp^2(1-p)^{k+s-2}\right)p(m-1,r,s)$$ Indeed, at a given step of the algorithm, say that we consider the vertex $u \in M$.
\begin{equation*}
    \begin{split}
        & p(m,r,s)\\
        & = Pr(\text{the algorithm fails given m,r,s}) \\
        & = \sum_{i \in \llbracket 1, r \rrbracket} Pr(\text{the algorithm fails given m,r,s} | N_{K\cup S}(u) = \{a_i,b_i\})Pr(N_{K\cup S}(u) = \{a_i,b_i\})\\
        & + Pr(\text{the algorithm fails given m,r,s} | \forall i \in \llbracket 1, r \rrbracket, N_{K\cup S}(u) \neq \{a_i,b_i\})\\
        & \times Pr(\forall i \in \llbracket 1, r \rrbracket, N_{K\cup S}(u) \neq \{a_i,b_i\})\\
        & = p(m-1,r-1,s+1)rp^2(1-p)^{k+s-2} + p(m-1,r,s)\left(1 - rp^2(1-p)^{k+s-2}\right)\\
        &\text{using} ~ Pr(\forall i \in \llbracket 1, r \rrbracket, N_{K\cup S}(u) \neq \{a_i,b_i\}) = 1 - \sum_{i \in \llbracket 1, r \rrbracket} Pr(N_{K\cup S}(u) = \{a_i,b_i\})
    \end{split}
\end{equation*}

We will now prove by induction that for all $m,r,s \in \mathbb{N} \text{~such that~} m \gs r$, $p(m,r,s) \ls r e^{-(m-r)p^2(1-p)^{k+r+s-2}}$. For the initialization, we need to prove that this is true for $p(r,r,s)$, for any $r$ and $s$, as well as for $p(m,0,s)$, for any $m$ and $s$. We have $p(m,0,s) = 0$, so we just need $p(r,r,s) \ls r$ for any $r \gs 1$, which is trivially true. Then, consider some $s \in \mathbb{N}, m,r \in \mathbb{N^*} \text{~such that~} m \gs r$, and suppose that for $p(m-1,r-1,s+1)$ and $p(m-1,r,s)$, the property is true. 

Then,
\begin{equation*}
    \begin{split}
        & p(m,r,s) \\
        & = rp^2(1-p)^{k+s-2}p(m-1,r-1,s+1) + \left(1 - rp^2(1-p)^{k+s-2}\right)p(m-1,r,s) \\
        & \ls rp^2(1-p)^{k+s-2}(r-1)e^{-(m-r)p^2(1-p)^{k+s+r-2}} \\
        &+ \left(1 - rp^2(1-p)^{k+s-2}\right)re^{-(m-r-1)p^2(1-p)^{k+s+r-2}}\\
        & = re^{-(m-r)p^2(1-p)^{k+s+r-2}}\left(p^2(1-p)^{k+s-2}(r-1) + \left(1 - rp^2(1-p)^{k+s-2}\right)e^{p^2(1-p)^{k+s+r-2}}\right)\\
        & \ls re^{-(m-r)p^2(1-p)^{k+s+r-2}}e^{p^2(1-p)^{k+s+r-2}}\left(p^2(1-p)^{k+s-2}(r-1) + \left(1 - rp^2(1-p)^{k+s-2}\right)\right)\\
        & = re^{-(m-r)p^2(1-p)^{k+s+r-2}}e^{p^2(1-p)^{k+s+r-2}}\left(1 - p^2(1-p)^{k+s-2}\right)\\
        & \ls re^{-(m-r)p^2(1-p)^{k+s+r-2}}e^{p^2(1-p)^{k+s+r-2}}e^{-p^2(1-p)^{k+s-2}}\\
        & \ls re^{-(m-r)p^2(1-p)^{k+s+r-2}} \text{~~~~ as $r \gs 0$}
    \end{split}
\end{equation*}

\noindent
At the end of the day, the probability that the algorithm fails is $p(m,r,0)$. So $1 - Pr(\exists S=\{u_1, \dots, u_r\}\subseteq M \text{~s.t.~} \forall i \in \llbracket 1, r \rrbracket, N_{K\cup S}(u_i) = \{a_i,b_i\}) \ls p(m,r,0) \ls r e^{-(m-r)p^2(1-p)^{k+r-2}}$.
\end{proof}

We are now ready to prove the existence of $k$-pairable graph states on a number of qubits cubic in $k$ (up to a logarithmic factor): 

\begin{proposition}
    \label{prop:pairability}
    For any constant $c>\frac{125e^2}{4} \approx 231$, there exists $k_0$ s.t.~for any $k>k_0$, there exists a $k$-pairable graph state on $\lfloor ck^3ln(k)^3\rfloor$ qubits.
\end{proposition}

\begin{proof}
For this proof we will generate a random graph, and we will apply \cref{lemma:unionbound} to show that with non-zero probability, such a graph satisfies the property of \cref{prop:cond_k_pairaibility}, i.e. for any $k$ disjoint pairs $\{a_1, b_1\},...,\{a_k, b_k\}$ of vertices, there exists a sequence of local complementations that maps $G$ to a graph $G'$ such that $\pi$ is an induced subgraph of $G'$, $\pi$ being the perfect matching with vertices $V_\pi=\bigcup_{i\in[k]}\{a_i,b_i\}$ and edges $E_\pi=\bigcup_{i\in[k]}\{(a_i,b_i)\}$. That will allow us to claim that the corresponding graph state is $k$-pairable, thus proving the proposition. Let $m \in \mathbb{N}$ and $p \in [0,1]$ to be fixed later on. The graph $G = (V,E)$, of order $n = m + 2k$, will be generated by the Erd\"os Rényi protocol: two vertices in $G$ are connected with probability $p$.

For some perfect matching $\pi$ represented by $\{\{a_1,b_1\},\dots,\{a_k,b_k\}\}$, the corresponding bad event $A_\pi$ is: "We cannot induce the perfect matching $\pi$ by means of local complementations on vertices from $V \setminus V_\pi$".
The first step of this proof is to bound $Pr(A_\pi)$. 
In order to be able to induce some subgraph on $V_\pi$, it is sufficient to find an independent set $S \subseteq V \setminus V_\pi$ of size ${2k \choose 2}$ such that for any $\{a,b\} \in E_\pi$, $\exists u_{a,b} \in S$  s.t.~$N_{V_\pi}(u_{a,b}) = \{a,b\}$. Indeed, one can induce any subgraph $H$ on $V_\pi$ by toggling each edge $(a,b)$ of $G[V_\pi]\Delta H$ by means of a local complementation on $u_{a,b}$. In our case, we do not always need to be able to flip all edges of $V_\pi$, because we only want to induce one particular subgraph 
(the perfect matching $\pi$). So the independent set that we look for might be of size less than ${2k \choose 2}$. Let's bound the probability of being able to construct such an independent set $S$ (which will give a direct bound for $Pr(A_\pi)$). We will use the fact that, as $p$ will be chosen to be small (not $1/2$), the number of edges to toggle will be linear in the vast majority of cases. We will be calling $r$ the number of edges to toggle in $V_\pi$. Let $t \in \mathbb{N}$ to be fixed later on. We will handle separately the cases where $r \ls t$ and $r \gs t +1$. 
$$ Pr(A_\pi) = Pr(r \ls t) .Pr(A_\pi | r \ls t) + Pr(r \gs t+1) .Pr(A_\pi | r \gs t+1)$$

Let's bound $Pr(r \ls t)$ and $Pr(A_\pi | r \gs t +1)$ by $1$, as they are supposed to be very close to 1.
$$ Pr(A_\pi) \ls Pr(A_\pi | r \ls t) + Pr(r \gs t+1)$$

Using \cref{lemma:probaS}, $Pr(A_\pi | r \ls t) \ls t e^{-(m-t)p^2(1-p)^{2k+t-2}}$.

To bound $Pr(r \gs t+1)$, we'll use some property of the binomial distribution. To simplify, we'll suppose we will always have to toggle the k edges corresponding to the pairs of $K$. Let's introduce a random variable $X$ that follows the distribution $B({2k \choose 2}-k,p)$. What we just said amounts to writing: $ Pr(r \gs t+1) \ls Pr(X \gs t+1-k)$. Indeed, the edges that are not the pairs of $K$ have to be removed, and some such edge exists with probability $p$. We'll use the Chernoff bound: With $\mu = \mathbb{E}[X] = p ({2k \choose 2}-k) = 2pk(k-1)$, for any $\delta > 0$, $Pr(X \gs (1+\delta)\mu) \ls e^{-\frac{\delta^2}{2+\delta}\mu}$. As we need $(1+\delta)\mu = t+1 - k$, we take $\delta = \frac{t+1-k-\mu}{\mu}$. The Chernoff bound then gives $$ Pr(X \gs t+1-k) \ls e^{-\frac{\left(\frac{t+1-k-\mu}{\mu}\right)^2}{\left(\frac{t+1-k+\mu}{\mu}\right)}\mu} = e^{-\frac{\left(t+1-k-\mu\right)^2}{\left(t+1-k+\mu\right)}}$$

At the end of the day, $ Pr(A_\pi) \ls t e^{-(m-t)p^2(1-p)^{2k+t-2}} + e^{-\frac{\left(t+1-k-2pk(k-1)\right)^2}{\left(t+1-k+2pk(k-1)\right)}}$.

The number of bad events is $d = \dfrac{n!}{k!(n-2k)!2^k}$. Using $k! \gs \dfrac{k^k}{e^{k-1}}$, it's upperbounded by $\dfrac{n^{2k} e^{k-1}}{(2k)^k} =\dfrac{1}{e}\left(\dfrac{n^2 e}{2k}\right)^k$. To apply \cref{lemma:unionbound}, we need  $dp_0 < 1$, where $p_0= t e^{-(m-t)p^2(1-p)^{2k+t-2}} + e^{-\frac{\left(t+1-k-2pk(k-1)\right)^2}{\left(t+1-k+2pk(k-1)\right)}}$ is the upper bound on the probability of the bad events.

It's sufficient to have: $$\textbf{(1)}~~  e^{-\frac{\left(t+1-k-2pk(k-1)\right)^2}{\left(t+1-k+2pk(k-1)\right)}}\left(\dfrac{n^2 e}{2k}\right)^k< \dfrac{e}{2} \text{~~and~~} \textbf{(2)}~~ \left(t e^{-(m-t)p^2(1-p)^{2k+t-2}}\right)\left(\dfrac{n^2 e}{2k}\right)^k< \dfrac{e}{2}$$

Let's show that these equations are satisfied for any large enough $k$ by choosing: $n = \left\lfloor c_2 k^3 \ln(k)^3 \right\rfloor$, $t = \left\lfloor c_1k \ln(k) \right\rfloor$ and $p = \frac{2}{2k+t}$ with $c_1 > 5$ and $c_2 > \frac{5e^2{c_1}^2}{4} > \frac{125e^2}{4} \approx 231$.

$\textbf{(1)}$: The equation translates to $k(2ln(n) + 1 -\ln(2k)) < \frac{\left(t+1-k-2pk(k-1)\right)^2}{\left(t+1-k+2pk(k-1)\right)} + \ln(e/2)$. $k(2ln(n) + 1 -\ln(2k)) = k(2ln(\lfloor c_2 k^3 \ln^3(k) \rfloor) + 1 -\ln(2k)) \ls k(2ln(c_2) +6ln(k) + 3\ln(\ln(k)) + 1 -\ln(k) -\ln(2)) \sim_{k\to \infty} 5kln(k)$. And $\frac{\left(t+1-k-2pk(k-1)\right)^2}{\left(t+1-k+2pk(k-1)\right)} + \ln(e/2) \sim_{k\to \infty} t = \lfloor c_1k \ln(k) \rfloor$. The choice of $c_1$ guarantees that for any large enough $k$, $\textbf{(1)}$ is satisfied.

$\textbf{(2)}$: Using \cref{lemma:technical1} (in the appendix), we get $e^{-(m-t)p^2(1-p)^{2k+t-2}} \ls e^{-(m-t)\frac{4e^{-2}}{(2k +t)^2}}$, so it's sufficient to prove $k(2ln(n) + 1 -\ln(2k)) < (m-t)\frac{4e^{-2}}{(2k+t)^2} - \ln(t) + \ln(e/2)$. We just saw that $k(2ln(n) + 1 -\ln(2k)) \sim_{k\to \infty} 5k \ln(k)$. Similarly, $m = n - 2k$ so $(m-t)\frac{4e^{-2}}{(2k+t)^2} - \ln(t) + \ln(e/2) \sim_{k\to \infty} c_2 k^3 \ln(k)^3 \frac{4e^{-2}}{(c_1k \ln(k))^2} = \frac{c_24e^{-2}kln(k)}{{c_1}^2}$. The choice of $c_2$ guarantees that for any large enough $k$, $\textbf{(2)}$ is satisfied.

Then, according to \cref{lemma:unionbound}, there exists a $k$-pairable graph state on $n$ qubits.
\end{proof}

\section{Vertex-minor universality}
\label{sec:ex_universal}

In the previous section, we considered graphs on which we can induce any perfect matching on any $2k$ vertices by means of local complementations. In this section we introduce the natural  combinatorial problem  of being able to induce any graph on any set of $k$ vertices:

\begin{definition}
    A graph $G$ is \emph{$k$-vertex-minor universal} if any graph on any $k$ vertices is a vertex-minor of $G$.
\end{definition}

The associated property on quantum states is the ability to induce not only EPR pairs, but any graph state on a given number of qubits by means of LOCC protocols. As any stabilizer state is equivalent to a graph state under local Clifford operations \cite{VandenNest04}, it leads to the following generalization of $k$-pairability:

\begin{proposition}
    If $G$ is a $k$-vertex-minor universal graph, then one can induce any stabilizer state of on any set of $k$ qubits in the corresponding graph state $\ket G$, by $CLOCC$ protocols.
\end{proposition}

Contrary to the pairability case (\cref{prop:cond_k_pairaibility}), the existence of CLOCC protocols to induce any stabilizer states on $k$ qubits from a given graph state $\ket G$ does not imply in general that $G$ is $k$-vertex-minor universal.~For instance, $K_2$ (graph with two vertices and one edge) is not $2$-vertex-minor universal since no local complementation can turn it into an empty graph. However, using CLOCC protocol (e.g. an X-measurement on each qubit), one can map the corresponding graph state (a maximally entangled pair of qubits) to the graph state composed a tensor product of two qubits.

Obviously, $2k$-vertex-minor universality  implies $k$-pairability: 

\begin{corollary}
    If $G$ is a $2k$-vertex-minor universal graph, then $\ket G$ is a $k$-pairable graph state.
\end{corollary}

As pointed out by Sergey Bravyi\footnote{Personal communication.}, a counting argument leads to an upper bound on the vertex-minor universality:

\begin{proposition}
If a graph $G$ of order $n$ is $k$-vertex-minor universal then $k< \sqrt{2n\log_2(3)}+2$. 
\end{proposition}

\begin{proof}
If $H$ of order $k$ is a vertex-minor of $G$ of order $n$, then $\ket H$ can be obtained from $\ket G$ by means of local Pauli measurements on $n-k$ qubits and  local Clifford unitaries on $k$ qubits. There are 3 possible Pauli measurements per qubit, so $3^{n-k}$ in total. Notice that for a fixed choice of Pauli measurements, different local Clifford transformations on the remaining $k$ qubits can only generate graph states which correspond to graphs that are equivalent up to local complementation. Moreover, it is known that there are at least $2^{\frac{k^2-5k}2-1}$ different graphs on $k$ vertices up to local complementation \cite{bahramgiri2007enumerating}. As a consequence, if $G$ is $k$-vertex-minor universal, we must have $3^{n-k}\gs 2^{\frac{k^2-5k}2-1}$. Using numerical analysis, this implies $k< \sqrt{2n\log_2(3)}+2$.
%
%
\end{proof}

Another upper bound, based on the  local minimum degree, can be obtained:

\begin{proposition}
    \label{prop:deltaloc_vmu}
    If a graph $G$ is $k$-vertex-minor universal then $k<\dloc(G)+2$.
\end{proposition}

\begin{proof}
By contradiction, assume there exists a graph $G$ that is $(\dloc(G)+2)$-vertex-minor universal.~$G$ contains a local set $L=D\cup\odd D$ of size $\dloc(G)+1$. By hypothesis, there exists a sequence of local complementations that maps $G$ to a graph $G'$ such that $G'[V(H)]=H$, where $H$ is the graph defined on $L \cup v$ ($v$ being an arbitrary vertex of $V \sm L$), such that its only edge is between $v$ and an arbitrary vertex $u \in L$.
Note that a local set is invariant by local complementation \cite{Perdrix06}. This means $L$ is still a local set in $G'$, i.e. there exists $D'$ s.t.~$L=D'\cup \odd{D'}$. If $u \in D'$ then $v \in Odd(D')$. If $u \notin D'$ then $u \notin Odd(D')$ because it's not connected to any other vertex of $L$. This contradicts that $L$ is still a local set.
\end{proof}

To illustrate the concept of vertex-minor universality, we provide minimal examples of $k$-vertex-minor universal graphs for $k$ up to $4$. The vertex-minor universality (and pairability) of the graphs below and in table \cref{fig:table} was observed by numerical analysis. The code consists of an exploration of the orbit of a given graph by local complementation (either by using essentially a breadth-first search, or by applying local complementations on random vertices, which yields results faster), to find the ${n \choose k}2^{k \choose 2}$ possible induced subgraphs of order $k$. 
\subparagraph{2-vertex-minor universality.}
As we saw above, $K_2$ is 1-pairable but not 2-vertex-minor universal: $K_3$ is actually the smallest 2-vertex-minor universal graph, because every edge can be toggled by a local complementation on the opposite vertex.
\subparagraph{3-vertex-minor universality.}
We observed that $C_6$ (the cycle of order 6) is a 3-vertex-minor universal graph, and there exist no 3-vertex-minor universal graphs of order smaller or equal to $ 5$. Indeed, using the database from \cite{Adcock20} along with \cref{prop:deltaloc_vmu}, it appears that the only graph of order smaller or equal to $ 5$ having a local minimum degree larger or equal to 2 (thus being a candidate for 3-vertex-minor universality) is $C_5$. We have observed however that $C_5$ is not 3-vertex-minor universal by exploring its orbit by local complementation composed of 132 graphs, and checking that none of them contains an independent set of size 3.
\subparagraph{4-vertex-minor universality.}
We observed that the 10-vertex "wheel" graph from \cite{DeWolf22} (see \cref{fig:wheel}) whose corresponding graph state has been proven to be 2-pairable, is also 4-vertex-minor universal.~Bravyi et al.~showed that no graph state with   $9$ or less qubits was 2-pairable using CLOCC protocols. Besides, any graph state corresponding to a 4-vertex-minor universal graph is 2-pairable using CLOCC protocols. Thus, there is no 4-vertex-minor universal graph of order smaller than $10$. We also observed that the Petersen graph (of order 10) is 4-vertex-minor universal (and thus 2-pairable). \\

These practical results, along with results for small Paley graphs, are summarized in \cref{fig:table}.

\begin{table}[h]
    
    \centering
    \scalebox{0.75}{
    \begin{tabular}{|c|c|c|c|c|c|c|c|c|c|}
    \hline
    graph $G$      & order & $\dloc(G)$ & 1-p? &2-vmu? &3-vmu?&2-p? & 4-vmu?& 5-vmu? & 3-p?\\ \hline
    $K_2$          & 2     & 1&yes & no &  &  &  &  &  \\ \hline
    $K_3$          & 3     & 2&yes & yes & no &  &  &  &  \\ \hline
    $C_6$          & 6     & 2&yes & yes & yes & no [cor \ref{cor:deltaloc}] & no [prop \ref{prop:deltaloc_vmu}] & no [prop \ref{prop:deltaloc_vmu}] & no [cor \ref{cor:deltaloc}]\\ \hline
    "Wheel" graph  & 10    & 3&yes & yes & yes & yes & yes & no [prop \ref{prop:deltaloc_vmu}] & no [cor \ref{cor:deltaloc}]\\ \hline
    Petersen graph & 10    & 3&yes & yes & yes & yes & yes & no [prop \ref{prop:deltaloc_vmu}]& no [cor \ref{cor:deltaloc}] \\ \hline
    13-Paley graph & 13    & 4&yes & yes & yes & yes & yes & no & no [cor \ref{cor:deltaloc}] \\ \hline
    17-Paley graph & 17    & 4&yes & yes & yes & yes & yes & yes & no [cor \ref{cor:deltaloc}]\\ \hline
    29-Paley graph & 29    & 10&yes & yes & yes & yes & yes & yes & yes \\ \hline
    \end{tabular}}
    \caption{A table summarizing the pairability of some graph states, and the vertex-minor universality of some graphs. "$k$-p?" is to be understood as "is the graph state $\ket G$ $k$-pairable?" and "$k$-vmu?" is to be understood as "is the graph  $G$ $k$-vertex-minor universal?". "yes" results were obtained by numerical analysis, by exploring the orbit of the graphs by local complementation, so we were able to check that each graph has the required vertex-minors. "no [$\cdot$]" results are direct applications of \cref{cor:deltaloc} and \cref{prop:deltaloc_vmu} (we had to compute the local minimum degree of each graph: this was done by brute force by computing every local set). $K_2$ not being 2-vertex-minor universal and $K_3$ not being 3-vertex-minor universal can be verified by checking the (very small) orbit by local complementation of these two graphs. The orbit by local complementation of the 13-Paley graph is, however, too big to be computed. To prove that this graph is not $5$-vertex-minor universal, we showed, using a slightly modified version of the program used in \cite{DeWolf22}, that no fully separable quantum state can be induced on the first 5 qubits by means of CLOCC protocols (which implies that the independent set on the first 5 vertices is not a vertex-minor of the graph).}
    \label{fig:table}
\end{table}

We prove the existence of an infinite family of $k$-vertex-minor universal graphs whose order is polynomial in $k$:

\begin{proposition} \label{prop:vmu} For any constant $c>\frac{3 e^2}{4}\approx 5.54$, there exists $k_0$ s.t.~for any $k>k_0$, there exists a  $k$-vertex-minor universal graph of order at most $ck^4\ln(k)$.
\end{proposition}

\begin{proof}
For this proof we will generate a random $(k+1)$-partite graph, and we will apply \cref{lemma:unionbound} to show that with non-zero probability, such a graph is $k$-vertex-minor universal.~Let $m \in \mathbb{N}$ to be fixed later on.  The graph $G = (V,E)$, of order $n = m (k+1)$, will be generated as follows. Consider $k+1$ sets of $m$ vertices $V_0$, $V_1$, $\dots$, $V_k$ that form a partition of $V$. Two vertices $u \in V_i$ and $v \in V_j$ in $G$ are not connected if $i=j$ (so any $V_i$ is an independent set), and are connected with probability $p \in [0,1]$ if $i \neq j$. 

For any $K \in {V \choose k}$, let $V^{(K)}\in \{V_0,\ldots , V_k\}$ s.t.~$K\cap V^{(K)} = \emptyset$. Notice that if for any pair $\{a,b\}\in {K\choose 2}$ there exists a vertex $u_{a,b} \in V^{(K)}$ s.t.~$N(u_{a,b})\cap K = \{a,b\}$ then one can induce any subgraph $H$ on $K$ by toggling each edge $(a,b)$ of $G[K]\Delta H$ by means of a local complementation on $u_{a,b}$. 

Given $K \in {V \choose k}$ and $\{a,b\}\in {K\choose 2}$, let $A_{K,a,b}$ be the bad event: ``$\forall u\in V^{(K)}$, $N_K(u)\neq \{a,b\}$''. Since the probability for a given $u\in V^{(K)}$ to satisfy $N_K(u)=\{a,b\}$ is $p^2(1-p)^{k-2}$ we have: 
\begin{eqnarray*}Pr(A_{K,a,b}) &=& \left(1-p^2(1-p)^{k-2}\right)^m\\
&\ls & e^{-mp^2(1-p)^{k-2}}
\end{eqnarray*}

We fix $p=\frac 2k$, as it minimizes $e^{-mp^2(1-p)^{k-2}}$ (see \cref{lemma:technical1} in the appendix), and get 
\[Pr(A_{K,a,b})\ls e^{-4e^{-2}mk^{-2}}\]

The number of bad events is $d= {m(k+1) \choose k} {k\choose 2}$ which is upperbounded by ${m(k+1)\choose k+1}$ when $m>  {k\choose 2}$  (see \cref{lemma:technical2} in the appendix). To apply \cref{lemma:unionbound}, we need $dp_0 < 1$, where $p_0= e^{-4e^{-2}mk^{-2}}$ is the upper bound on the probability of the bad events. 
\begin{eqnarray*}dp_0&\ls & {m(k+1)\choose k+1}e^{-4e^{-2}mk^{-2}}\\
&\ls &  2^{m(k+1)H(\frac 1 m)}e^{-4e^{-2}mk^{-2}}\\
&=&e^{m(k+1)H(\frac 1 m)\ln(2)-4e^{-2}mk^{-2}}
\end{eqnarray*}
where $H(x) = -x\log_2(x)-(1-x)\log_2(1-x)$ is the binary entropy. 

By choosing $m=\left\lfloor c\frac{k^4}{k+1}\ln(k) \right\rfloor$ with $c > \frac{3 e^2}{4} \approx 5.54$, we have $m(k+1)H(\frac 1 m)\ln (2) \sim_{k\to \infty} 3k\ln (k)$ and $4e^{-2}mk^{-2} \sim_{k\to \infty} 4e^{-2}ck\ln (k)$. The choice of $c$ guaranteed that for any large enough $k$, $$e^{m(k+1)H(\frac 1 m)\ln(2)-4e^{-2}mk^{-2}}< 1,$$ and thus, according to \cref{lemma:unionbound}, there exists a graph of order $m(k+1)$ which is $k$-vertex-minor universal.
\end{proof}

\section{Robust pairability}
\label{sec:robust}
In this section, we explore a natural extension of pairability in the presence of errors or malicious parties. We say that a $k$-pairable state is $m$-robust if one can create any $k$ pairs of maximally entangled qubits, independently of the actions of any set of $m$ identified malicious parties:

\begin{definition}
A $k$-pairable state $\ket \psi$ on a set $V$ of qubits is \emph{$m$-robust}  if for any set $M\subseteq V$ of size at most $m$ and any matching $\pi$ of size $k$ on the vertices $V\setminus M$, there is an LOCC protocol that transforms $\rho =Tr_M(\ket{\psi}\bra \psi)$ into  $\ket \pi$. 
\end{definition}

Vertex-minor universal graphs provide robust pairable quantum states: 

\begin{proposition}
If $G$ is a $k$-vertex-minor universal graph then for any $k'\ls \frac k2$, $\ket G$ is a $(k-2k')$-robust $k'$-pairable state. 
\end{proposition}

\begin{proof}
Let  $G$ be a $k$-vertex-minor universal  graph, $M$ a set of at most $k-2k'$ vertices and $K$ a (disjoint) set of $2k'$ vertices. We consider any matching $\tau$ of size $k'$ on the vertices $K\cup M$, such that each vertex of $K$ is of degree 1, and those of $M$ are of degree $0$. It is enough to show that there is an LOCC protocol transforming $Tr_M(\ket{G}\bra G)$ into $\ket {\tau[K]}$.

As $\tau$ is of order $|K|+|M|\ls k$, $\tau$ is a vertex-minor of $G$,  hence there exists a sequence of local complementations transforming $G$ into $G'$ such that $G'[K\cup M]=\tau$. In terms of graph state, it means that there are local (Clifford) unitaries $U_M$ (acting on the qubits of $M$) and $U_{\overline M}$ (acting on the other  qubits) s.t.~$U_MU_{\overline M} \ket G = \ket {G'}$. 
Notice that applying $U_{\overline M}$ on $\rho= Tr_M(\ket G\bra G)$ leads to the state $\rho' = U_{\overline M}Tr_M(\ket{G}\bra {G})U^\dagger_{\overline M} = Tr_M(U_{\overline M}U_M\ket{G}\bra {G}U^\dagger _{\overline M}U^\dagger_M) =Tr_M(\ket{G'}\bra {G'})$. Thus, from now on, we assume that the overall state is $\ket{G'}$, and we show that $\ket{G'}$ can be turned into $\ket{\tau[K]}$ without the help of the parties in $M$.

Vertex deletions can be implemented by means of standard basis measurements and local corrections,  more precisely, to transform $\ket{G'}$ into $\ket{G'\setminus u}$, one can measure $u$ leading to the state $Z_{N(u)}^{s_u}\ket{G'\setminus u}$ where $s_u\in\{0,1\}$ is the classical outcome of the measurement. As a consequence, the measurements of the qubits $V\setminus (K\cup M)$ of $\ket {G'}$ lead to the state $\ket{G'[K\cup M]} = \ket{\tau}$ up to some $Z$ corrections which depends on the classical outcomes of the measurements, thus the state is before any classical communications and corrections, of the form $Z_AZ_B\ket \tau$ for some subsets $A\subseteq K$, $B\subseteq M$. The parties of $V\setminus (K\cup M)$ send their classical outcomes to the parties of $K$ so that the correction $Z_A$ can be applied leading to the state $Z_B\ket \tau$. As the qubits of $M$ are separable for those of $K$ in $\ket \tau$ (the only edges of $\tau$ are inbetween vertices of $K$), tracing out the qubits of $M$ in the state $Z_B\ket \tau$ leads to $\ket {\tau[K]}$.
\end{proof}

\section{Conclusion}

We showed here that there exist polynomial size $k$-pairable quantum states and provided new upper bounds for $k$-pairability. We also introduced a new related combinatorial notion called vertex-minor universality, for which we gave similar properties, showing the existence of polynomial size graphs that are $k$-vertex-minor universal and providing lower bounds for $k$-vertex-minor universality based on the minimum degree by local complementation. We also provided minimal examples of $k$-vertex-minor universal graphs for small values of $k$. 
Finally, we initiated the study of  a robust version of $k$-pairability, in the presence of errors or malicious parties.  

This leaves open some questions for future work. 

\begin{itemize}
\item Our proof for the existence of polynomial-size $k$-pairable quantum states is non-constructive, and their size might be far from optimal. 
Explicit constructions of polynomial-size $k$-pairable quantum states in the same manner as the Reed-Muller CSS states from Bravyi et al., would be the logical next step. Paley graphs, thanks to their `large' local minimum degrees, are candidates to provide good $k$-pairable states,  using potentially even less qubits than those exhibited by non-constructive probabilistic methods in this paper. Similar questions also naturally apply to the explicit constructions of $k$-vertex-minor universal graphs.

\item Pairability of graph states, when restricted to CLOCC protocols, is fully characterized by the combinatorial properties of the associated graph (see  \cref{prop:cond_k_pairaibility}). Does it extend to pairability with arbitrary LOCC protocols? Does there exist a $k$-pairable graph state which is not $k$-pairable by means of CLOCC protocols?  

\item Even though $2k$-vertex-minor universality is a stronger requirement than $k$-pairability, it is not clear whether there exist $k$-pairable graph states on more than $2$ vertices whose underlying graphs are not $2k$-vertex-minor universal.

\item The $k$-pairable graph states that exist with non-zero probability - as shown in the proof of \cref{prop:pairability} - satisfy the following property: in order to create $k$ particular EPR-pairs, one can compute the local operations to be applied on the graph state (or, equivalently, the local complementation to be applied on the underlying graph: there are $O(k^2)$ of them) in polynomial time. This is a consequence of the fact that the independent set in the proof is found using a greedy algorithm. This raises the question of which $k$-pairable states possess this property. Furthermore, we could study how to generate graphs where the greedy algorithm works for every EPR-pair with high probability (whereas in this work we only show the existence of one such graph). A similar discussion is pertinent in the case of $k$-vertex-minor universal graphs.
\end{itemize}

\section*{Acknowledgements}

We thank Sergey Bravyi and Ronald de Wolf for fruitful comments on an early version of this paper. This work is supported by the PEPR integrated project EPiQ ANR-22-PETQ-0007 part of Plan France 2030, by the STIC-AmSud project Qapla’ 21-STIC-10, by the QuantERA grant EQUIP ANR-22-QUA2-0005-01, and by the European projects NEASQC and HPCQS.


\bibliographystyle{plainurl}
\bibliography{ref}

\section*{Appendix: Technical lemmas}
\label{app:technical}

\begin{lemma}
    Given $c > 0$, the value of $x$ that maximizes $f: x \rightarrow x^2(1-x)^c$ is $x = \frac{2}{c+2}$. Also, $f(\frac{2}{c+2}) \ls \frac{4e^{-2}}{(c+2)^2}$. \label{lemma:technical1}
\end{lemma}

\begin{proof}
$f'(x) = 2x(1-x)^c - x^2(1-x)^{c-1}c = x(1-x)^{c-1}(2(1-x)-cx) = x(1-x)^{c-1}(2 - (c+2)x)$. We note that $f'(x) = 0$ when $x = \frac{2}{c+2}$. Also, $f(x) \ls x^2e^{-\frac{xc}{1-x}}$, and $\frac{\frac{2}{c+2}c}{1- \frac{2}{c+2}} = \frac{\frac{2}{c+2}c}{\frac{c}{c+2}} = 2$, so $f(\frac{2}{c+2}) \ls \frac{4e^{-2}}{(c+2)^2}$.
\end{proof}

\begin{lemma}
    \label{lemma:technical2}
For any $k>1$, any $m> {k\choose 2}$, \[{m(k+1) \choose k} {k\choose 2}\ls {m(k+1)\choose k+1}\]
\end{lemma}
    
\begin{proof}
\begin{equation*}
    \begin{split}
        \frac{{m(k+1)\choose k+1}}{{m(k+1) \choose k} {k\choose 2}}& = \frac{(m(k+1))!}{(m(k+1)-k-1)!(k+1)!}\frac{(m(k+1)-k)!k!  } {(m(k+1))!}  \frac{2! (k-2)!}{k!}
        =\frac{2(m(k+1)-k)}{(k+1)k(k-1)}\\
    &\gs\frac{2\left(\left({k\choose2}+1\right)(k+1)-k\right)}{(k+1)k(k-1)}
    = \frac{2\left({k\choose2}(k+1)+1\right)}{(k+1)k(k-1)}
    =\frac{(k-1)k(k+1)+2}{(k+1)k(k-1)}\\
    &=1+\frac{2}{(k+1)k(k-1)}
    \gs 1
    \end{split}
\end{equation*}
\end{proof}


\end{document}